\pdfoutput=1

\documentclass[11pt]{article}

\usepackage{amsmath}
\usepackage{amsfonts}
\usepackage{amssymb}
\usepackage{amsthm}
\usepackage[margin=1in]{geometry}
\usepackage[linesnumbered,algoruled,noend]{algorithm2e}
\usepackage{paralist}

\usepackage[colorlinks,linkcolor=blue,filecolor=blue,citecolor=blue,urlcolor=blue,pdfstartview=FitH,hypertexnames=false]{hyperref}

\newtheorem{theorem}{Theorem}[section]
\newtheorem{lemma}[theorem]{Lemma}

\newtheorem{corollary}[theorem]{Corollary}

\newcommand{\N}{\mathbb{N}}
\newcommand{\R}{\mathbb{R}}
\newcommand{\BO}{\mathcal{O}}

\newcommand{\namedref}[2]{\hyperref[#2]{#1~\ref*{#2}}}
\newcommand{\sectionref}[1]{\namedref{Section}{#1}}

\newcommand{\theoremref}[1]{\namedref{Theorem}{#1}}

\newcommand{\lemmaref}[1]{\namedref{Lemma}{#1}}

\newcommand{\corollaryref}[1]{\namedref{Corollary}{#1}}

\newcommand{\algref}[1]{\namedref{Algorithm}{#1}}

\newcommand{\lineref}[1]{\namedref{Line}{#1}}
\newcommand{\equalityref}[1]{\hyperref[#1]{Equality~\eqref{#1}}}
\newcommand{\inequalityref}[1]{\hyperref[#1]{Inequality~\eqref{#1}}}

\begin{document}
\setcounter{tocdepth}{3}

\title{Node-Initiated Byzantine Consensus Without a Common Clock}

\author{Danny Dolev, Hebrew University of Jerusalem, dolev@cs.huji.ac.il\\
Christoph Lenzen, Massachusetts Institute of Technology, clenzen@csail.mit.edu}

\date{}

\maketitle

\begin{abstract}
The majority of the literature on consensus assumes that protocols are jointly
started at all nodes of the distributed system. We show how to remove this
problematic assumption in semi-synchronous systems, where messages delays and
relative drifts of local clocks may vary arbitrarily within known bounds.
Our framework is self-stabilizing and efficient both in terms of communication
and time; more concretely, compared to a synchronous start in a synchronous
model of a non-self-stabilizing protocol, we achieve a constant-factor increase
in the time and communicated bits to complete an instance, plus an additive
communication overhead of $\BO(n\log n)$ broadcasted bits per time unit and
node. The latter can be further reduced, at an additive increase in time
complexity.
\end{abstract}

%
%

\section{Introduction}
Consensus is a fundamental fault-tolerance primitive in distributed systems,
which has been introduced several decades ago~\cite{pease80}. Given a system of
$n$ nodes, some of which may be faulty and disobey the protocol in an arbitrary
fashion, the problem can be concisely stated as follows. Each node $v$ is given
some input $i(v)$ from some set of possible inputs $I$. A consensus protocol
resilient to $f$ faults must---under the constraint
that at most $f$ nodes are faulty---satisfy that
\begin{compactitem}
\item[\textbf{Termination:}] every  correct node eventually terminates
and outputs a value $o(v)\in I$;
\item[\textbf{Agreement:}] $o(v)=o(w)$ for correct nodes $v,w$ (we thus may talk
of \emph{the} output of the protocol);
\item[\textbf{Validity:}] if $i(v)=i(w)$ for all correct $v,w$, this is also the
output value.
\end{compactitem}
The main optimization criteria are the \emph{resilience} $f$ (\cite{pease80}),
the \emph{running time}, i.e., the time until all nodes terminate
(\cite{Fischer1982}), and the number of messages and bits sent by correct nodes
(\cite{dolev85}).

\subsection*{A Motivating Example}

Two major international banks, $A$ and $B$, are long-standing rivals. The evil
CEO of bank $B$ hires a professional infiltration specialist and hacker to cause
maximal damage to bank $A$. She offers her services to bank $A$ to inspect the
robustness of their systems. Having full access, she learns that bank $A$ has
been very thorough: All data is replicated at multiple locations, and any
changes are done using a state-of-the-art consensus algorithm; if a customer or
employee commits any value, it is encrypted and sent to the different locations
via secure channels, where it serves as input for the consensus routine.

Consequently, she looks for a way to bring the system down. The consensus
algorithm is Byzantine fault-tolerant, i.e., resilient to arbitrary behavior of
a minority of the data centers. However, she knows that such algorithms are
costly in terms of computation and communication. Closer examination reveals
that the bank solved this by restricting the frequency at which consensus is run
and using a synchronous protocol. Every second, an instance of the algorithm is
started that commits the batch of recent changes. Due to these choices, the
infrastructure is capable of dealing with an influx of operations well beyond
the peak loads.

Nonetheless, the spy now knows how to break the system. She takes control of the
external sources the data centers obtain their timing information from---which
are outside bank $A$'s control---and feeds conflicting time values to the
different data centers. For most operations, the result is denial of service,
but in some instances the lack of synchrony leads to inconsistent commits.
Despite bank $A$'s excellent IT security, the damage is dramatic.

\subsection*{An Alternative Approach}

Could bank $A$ have averted its fate? Clearly, relying on a trusted external
time reference is chancy. On the other hand, dropping timing conditions entirely
would necessitate to use asynchronous consensus protocols, which offer
substantially worse trade-offs between efficiency and reliability. A third
option is to make use of a time reference under control of the bank. However, to
avoid introducing a new point of vulnerability to the system, it needs to be
replicated as well. An obvious choice here is to equip each \emph{node} of
the system (a.k.a.\ data center) with its own clock.

There are two possibilities for leveraging such \emph{hardware clocks}: (i)
using local timing conditions at each node in a consensus algorithm, or (ii) running a
clock synchronization algorithm to compute synchronized logical clocks and
executing synchronous algorithms driven by these clocks.

Regarding (ii), repetitive consensus (or other, tailored solutions, e.g.,
\cite{srikanth87}) can be used to maintain a common clock despite faults.
However, how can initial synchronization be established, without relying on some
external means? Moreover, if a node rejoins the network (after maintenance or a
transient fault), or the system is to recover from a partition, is this possible
without external intervention? The notion of \emph{self-stabilization}
\cite{dijkstra74} covers all these scenarios: a self-stabilizing algorithm must
eventually establish a valid state, no matter the initial state it is started
from.

Byzantine self-stabilizing clock synchronization
algorithms~\cite{dolev11,dolev04,fuegger13} provide solutions for (ii).
There is a close relation to consensus in general and (i) in particular. To the
best of our knowledge, all known algorithms for (ii) employ techniques commonly
found in consensus algorithms, and most of them make explicit use of consensus
protocols as subroutines.

With respect to (i), the question arises on how to decide on \emph{when} to run
consensus. Without any agreement on a global time among the nodes, some other
global reference needs to be established in order to jointly start an instance
of the consensus protocol at all nodes. While a broadcast of a single node could
establish such a reference, this is problematic if the respective node is
corrupted. Even if we add redundancy by allowing for multiple nodes initiating
instances, a mechanism is required to prevent that corrupt nodes overwhelm the
system with too many instances or initialize runs inconsistently.

In particular the latter issue is not be taken lightly, as it entails to
establish agreement among the nodes on whether consensus is to be run or not!
Considering that the vast majority of the literature on consensus assumes that
all nodes in unison start executing an algorithm at a specific
time\footnote{Observe that even asynchronous models assume that all nodes
fully participate at the invoked consensus once they ``wake-up''. This can be
logically mapped to all nodes participating unconditionally from time $0$ on,
much like the unison start of the synchronous algorithms.} we consider it both
surprising and alarming that this issue is not addressed by existing work.

\subsection*{Further Related Work}

The issues that surface when running consensus in practice have been studied
extensively. Researchers distinguished between ``Atomic Broadcast'' that may be
repeatedly executed and ``Consensus'' (\cite{PonzioStrong}) that was considered
as a ``single shot''. Running a synchronous protocol in a semi-synchronous
environment was studied
comprehensively~\cite{Attiya20011096,Attiya:1991,222994,Cristian1995158,Dwork:1988,Ponzio:1991,Herlihy:1998}.
Lower bounds, upper bounds and failure models were presented, and complexity
measures were analyzed. But all previous work explicitly or implicitly assumes
that when consensus is invoked, every node has an input and executes the
protocol to reach the target value that is determined by the set of inputs.
Solving the question we consider using previous work translates to continuously
running consensus on whether to run consensus or not. Thus, such an approach
enables the faulty nodes to cause the correct nodes to get involved in an
unbounded number of invocations.

\subsection*{Contribution}

In this work, we provide a generic solution to (i), where the hardware clock
$H_v$ of node $v$ may run at rates that vary arbitrarily within $[1,\vartheta]$
and message delays may vary arbitrarily within $(0,d)$.
\begin{theorem}\label{theorem:main}
Suppose ${\cal P}$ is a synchronous consensus protocol that tolerates $f<n/3$
faults, runs for $R\in \operatorname{polylog}(n)$ rounds, and guarantees that no
correct node sends more than $B$ bits. Then for each $T\geq 2\vartheta^2d$,
there are a value $S\in \BO(R+T)$ and an algorithm with the following
properties.
\begin{compactitem}
  \item Each correct node $v$ can initiate an instance of ${\cal P}$ at any
  time $t\geq S$, provided that it has not done so at any time $t'<t$ for which
  $H_v(t)-H_v(t')\leq T$.
  \item For any instance that terminates at a time larger than $S$, it holds
  that nodes determine their inputs according to their local view of the system
  during some interval $[t_1,t_1+\BO(1)]$, and terminate during some interval
  $[t_2,t_2+\BO(1)]$, where $t_2\in t_1+\Theta(R)$.
  \item If a correct node initiates an instance at time $t\geq S$, then $t_1=t$.
  \item Each instance for which $t_2\geq S$ satisfies termination, agreement,
  and validity.
  \item Each correct node sends at most $\BO(n^2\log n+nBR/T)$ amortized bits
  per time unit.
  \item The above guarantees hold in the presence of $f$ faulty nodes and
  for arbitrary initial states.
\end{compactitem}
\end{theorem}
This statement can be extended to randomized algorithms that satisfy agreement
and validity \emph{with high probability},\footnote{I.e., probability at least
$1-1/n^c$, where $c$ is an arbitrary constant that is chosen upfront.} and
accepting that inputs are determined within a time window of size $\BO(T)$
enables to decrease the amortized bit complexity per node to $\BO((n^2\log
n+nBR)/T)$ (see \sectionref{app:further}). We remark that our results bear the
promise of improved solutions to (ii).

\section{Model and Problem}

For the purpose of our analysis, we assume that there is a global
\emph{reference time}. Whenever we talk of a time, it will refer to this
reference time from $\R_0^+$, which is unknown to the nodes.

\subsection*{Distributed System}

We model the distributed system as a finite set of $n$ nodes $V=\{1,\ldots,n\}$
that communicate via message passing. Each such message is subject to a delay
from the range $(0,d)$, where $d\in \BO(1)$. Every node can directly communicate
to every other node. The sender of a message can be identified by the receiver.
Up to $f<n/3$ nodes are \emph{Byzantine faulty}, i.e., exhibit arbitrary
behavior. We denote the set of \emph{correct} nodes (i.e., those that are not
faulty) by $G$. Initially, correct nodes' states are arbitrary, and the
communication network may deliver arbitrary messages prior to time $d$.

Each node $v\in V$ is equipped with a hardware clock $H_v:\R^+_0\to \R^+_0$.
(We will show in \sectionref{sec:self-stab} that bounded and discrete clocks are
sufficient, but use the unbounded and continuous abstraction throughout our
proofs.) Clock rates are within $[1,\vartheta]$ (with respect to the reference
time), where $\vartheta-1$ is the \emph{(maximal) clock drift}. For any times
$t<t'$ and correct node $v\in G$, it holds that $t'-t\leq H_v(t')-H_v(t)\leq
\vartheta (t'-t)$. Since hardware clocks are not synchronized, nodes typically
use them to approximately measure timespans by \emph{timeouts}. A timeout can be
either \emph{expired} or \emph{not expired}. When node $v$ \emph{resets} a
\emph{timeout of duration $T\in \R^+$} at time $t$, it is not expired during
$[t,t')$, where $t'$ is the unique time satisfying that $H_v(t')-H_v(t)=T$. The
timeout \emph{expires} when it has not been reset within the last $T$
\emph{local} time (i.e., the last $T$ units of time according to $H_v$).

\textbf{Algorithms and Executions.} Executions are event-based. An \emph{event}
can be a node's hardware clock reaching a certain value, a timeout expiring, or
the reception of a message. Upon an event at time $t$, a node may read its
hardware clock $H_v(t)$, perform local computations, store values, send
messages, and reset timeouts. For simplicity, all these operations
require $0$ time and we assume that no two events happen at the same time. 

\subsection*{Problem Formulation}

We will solve a slightly weaker problem than stated in the abstract, from which
the claimed results readily follow. We are given a deterministic, synchronous,
binary $R$-round consensus protocol ${\cal P}$ resilient to $f<n/3$ faults. The
goal of the \emph{initiation problem} is to enable correct nodes to initiate
independent executions of instances of ${\cal P}$. More precisely:
\begin{compactenum}
\item Each instance carries a label $(v,H_v)\in V\times \R^+_0$.
\item If node $v\in G$ decides to \emph{initiate} an instance, this instance has
label $(v,H_v(t))$.
\item For each instance, each node $w\in G$ decides whether it
\emph{participates} in the instance at some time $t_w$. If it does, we assume
that it has access to some appropriate input $i_w(v,H_v,t_w)\in \{0,1\}$ (which
it may or may not use as the input for this instance), and we require that it
eventually terminates the instance and outputs some value $o_w(v,H_v)\in
\{0,1\}$.
\item If $v\in G$ initializes instance $(v,H_v(t))$ at time $t$, each $w\in G$
decides to participate, $t_w\in t+\Theta(d)$, and it will terminate this
instance at some time from $t+\Theta(R)$.
\item If $w,w'\in G$ participate in instance $(v,H_v)$, then
$o_w(v,H_v)=o_{w'}(v,H_v)$.
\item If all nodes in $G$ participate in an instance and use the same input
$b$, they all output $b$.
\item If $o_w(v,H_v)\neq 0$ for some $w\in G$, then all nodes in $G$
participate in this instance, terminate within a time window $[t^-,t^+]$ of
constant size, and some $u\in G$ satisfies that $i_u(v,H_v,t_u)=o_w(v,H_v)$ for
some $t_u\in t^+-\Theta(R)$.
\end{compactenum}
Compared to ``classic'' consensus, property 4 corresponds to
\emph{termination}, property 5 to \emph{agreement}, and property 6 to
\emph{validity}. Note that validity is replaced by a safety property in case not
all correct nodes participate: property 7 states that non-zero output is
feasible only if no correct node is left out. Finally, property 8 makes sure
that all nodes participate in case a non-faulty node initializes an instance,
therefore ensuring validity for such instances.

Put simply, these rules ensure that each instance $(v,H_v(t))$ initiated by a
correct node behaves like a ``classic'' consensus instance with inputs
$i_w(v,H_v(t),t_w)$, where $t_w\approx t$, which terminates within $\Theta(R)$
time, and roughly simultaneously at all correct nodes. If a faulty node
initializes an instance, the timing conditions are guaranteed only if some
non-faulty node outputs a non-zero value; in this case we are also ensured that
there has been some corresponding input $\Theta(R)$ time in the past, i.e., the
computed output is valid.

Assuming that the ``fallback'' output $0$ never causes any damage, this is
certainly acceptable. In particular, we can use the initiation problem to agree
on whether an instance of an arbitrary (possibly non-binary) consensus protocol
should be jointly started by all correct nodes upon terminating the instance, at
roughly the same time. Setting the inputs $i_v\equiv 1$ (in the initiation
problem) will then ensure that the result is a solution to the task stated in
the abstract and the introduction.

\section{Algorithm}\label{sec:algo}

In this section, we present a simplified version of our algorithm; in
particular, it is not self-stabilizing and has unbounded communication
complexity. We will address these issues in \sectionref{sec:self-stab}.

Our algorithm consists of three main components beside the employed consensus
protocol ${\cal P}$. The first provides each node with an estimate of each other
node's clock, with certain consistency guarantees that apply also to clocks of
faulty nodes. The second uses this shared timing information to enforce clean
initialization of consensus instances within a small window of time. Finally,
the third component provides a wrapper for the consensus protocol that simulates
synchronous execution. Before presenting the protocol, however, we need to
introduce an additional property the employed consensus protocol ${\cal P}$ must
satisfy.

\subsection*{Silent Consensus}
We call a consensus protocol \emph{silent}, if in any execution in which all
correct nodes have input $0$, correct nodes send no messages and output $0$.
Observe that even if not all correct nodes participate, a silent protocol will
end up running correctly and output $0$ at all participating nodes if no correct
node has non-zero input. We show that any consensus protocol can be
transformed into a silent one.
\begin{lemma}\label{lemma:silent}
Any synchronous consensus protocol ${\cal P}$ in which nodes sent at most
$B$ bits can be transformed into a silent synchronous binary consensus protocol
${\cal P}_{\!s}$ with the same properties except that it runs for two more
rounds, during which each node may perform up to $2$ $1$-bit broadcasts.
\end{lemma}
\begin{proof}
The new protocol can be seen as a ``wrapper'' protocol that manipulates the
inputs and then each node may or may not participate in an instance of the
original protocol. The output of the original protocol, ${\cal P}$, will be
taken into account only by correct nodes that participate throughout the
protocol, as specified below. In the first round of the new protocol, ${\cal
P}_{\!s}$, each participating node broadcasts its input if it is not $0$ and
otherwise sends nothing. If a node receives fewer than $n-f$ times the value
$1$, it sets its input to $0$. In the second round, the same pattern is applied.

Subsequently, ${\cal P}$ is executed by all nodes that received at least $f+1$
messages in the first round (where any missing messages from nodes that do not
participate are set to an arbitrary valid message by the receiver). If in the
execution of ${\cal P}$ a node would have to send more bits than it would have
according to the known bound $B$, it (locally) aborts the execution
of ${\cal P}$. Likewise, if the running time bound of ${\cal P}$ would
be violated, it aborts as well. Finally, a node outputs $0$ in the new protocol
if it did not participate in the execution of ${\cal P}$, aborted it, or
received $f$ or less messages in the second round, and it outputs the result
according to the run of ${\cal P}$ otherwise.

We first show that the new protocol, ${\cal P}_{\!s}$, is a consensus protocol
with the same resilience as ${\cal P}$ and the claimed bounds on communication
complexity and running time. We distinguish two cases. First, suppose that all
correct nodes participate in the execution of ${\cal P}$ at the beginning of the
third round. As all nodes participate, the bounds on resilience, communication
complexity, and running time that apply to ${\cal P}$ hold in this execution,
and no node will quit executing the protocol before termination. To establish
agreement and validity, again we distinguish two cases. If all nodes output the
outcome of the execution of ${\cal P}$, these properties follow right away since
${\cal P}$ satisfies them; here we use that although the initial two rounds
might affect the inputs of nodes, a node will change its input to $0$ only if
there is at least one correct node with input $0$. On the other hand, if some
node outputs $0$ because it received $f$ or less messages in the second round of
${\cal P}_{\!s}$, no node received more than $2f<n-f$ messages in the second
round. Consequently, all nodes executed ${\cal P}$ with input $0$ and computed
output $0$ by the agreement property of ${\cal P}$, implying agreement and
validity of the new protocol.

The second case is that some correct node does not participate in the execution
of ${\cal P}$. Thus, it received at most $f$ messages in the first round of
${\cal P}_{\!s}$, implying that no node received more than $2f<n-f$ messages in
this round. Consequently, correct nodes set their input to $0$ and will not
transmit in the second round. While some nodes may execute ${\cal P}$, all
correct nodes will output $0$ no matter how ${\cal P}$ behaves. Since nodes
abort the execution of ${\cal P}$ if the bounds on communication or time
complexity are about to be violated, the claimed bounds for the new protocol
hold.

It remains to show that the new protocol is silent. Clearly, if all correct
nodes have input $0$, they will not transmit in the first two rounds. In
particular, they will not receive more than $f$ messages in the first round and
not participate in the execution of ${\cal P}$. Hence correct nodes do not
send messages at all, as claimed.
\end{proof}

\begin{algorithm}[t]\label{algo:clock}
\If{$H_v(t)\!\!\mod 2\vartheta d = 0$}{
  \For{$w\in\{1,\ldots,n\}\setminus \{v\}$}{
    \If{$H_v(t)-R_v^w>(2\vartheta^2+\vartheta) d$}
      {$M_{vww}:=\bot$}\nllabel{line:too_slow}
  }
  broadcast update$(M_{vv})$\nllabel{line:update}
}
\If{received \emph{update}$(M_{w11},\ldots,M_{wnn})$ from node $w$ at time $t$}{
  \If{$H_v(t)-R_v^w < d$ or $M_{www}-M_{vww}\neq 2\vartheta
    d$\nllabel{line:too_fast}}{$M_{vww}:=\bot$}
  \Else{$M_{vw}:=(M_{w11},\ldots,M_{wnn})$}
  \If{$|\{u\in V\,|\,|M_{vww}-M_{vuw}|\leq (2\vartheta^2+4\vartheta)
    d\}|<n-f$}{$M_{vww}:=\bot$}\nllabel{line:check_support}
  \If{$M_{vww}=\bot$}{$R_v^w:=\bot$}
  \Else{$R_v^w:=H_v(t)$}\nllabel{line:reset_R}
}
\caption{Actions of node $v\in V$ at time $t$ that relate to maintaining clock
estimates.}
\end{algorithm}

\subsection*{Distributing Clocks}

The first step towards simulating round-based protocols is to establish a common
timeline for each \emph{individual} node. This can be easily done by a
broadcast, however, such a simple mechanism would give the adversary too much
opportunity to fool correct nodes. Therefore, we require nodes to continuously
broadcast their local clock values, and keep updating the other nodes on the
values they receive. By accepting values only if a sufficient majority of nodes
supports them, we greatly diminish the ability of faulty nodes to introduce
inconsistencies. In addition, nodes check whether clock updates occur at a
proper frequency, as it is known that correct nodes send them regularly. This
approach is more robust than the timing tests used in~\cite{DHSS95} to eliminate
untimely messages sent by faulty nodes.

Nodes exchange their clock values at a regular frequency and relay to  others
the values they have received. To this end, node $v$ maintains memory entries
$M_{vuw}$, $w,u\in V$, where $M_{vuw}=H$ is to be understood as ``$u$ told me
that $w$ claimed to have clock value $H$''. At any time $t$, node $v$ will
either trust the clock value node $w$ claims to have, i.e., its estimate of
$H_w(t)$ is $M_{vww}(t)$, or it will not trust $w$. The latter we express
concisely by $M_{vww}(t)=\bot$, i.e., any comparison involving $M_{vww}(t)$ will
fail. Note that if that happens at any time, it cannot be undone; since $w$
proved to be faulty and we are not concerned with self-stabilization here, $v$
will just ignore $w$ in the future. For simplicity, we set $M_{vvv}(t):=H_v(t)$
for all times $t$ and, to avoid initialization issues, assume that
$M_{vuw}(0)=H_w(0)$ for all $u,v,w\in G$.\footnote{Since the full algorithm is
self-stabilizing, we do not need to worry about initialization in our simplified
setting.} Finally, $v$ stores its local time when it received a clock update
from $w$ in the variable $R_v^w$, in order to recognize $w$ violating the timing
constraints on update messages. The actions $v$ takes in order to maintain
accurate estimates of other's clocks are given in \algref{algo:clock}. The
``broadcast'' in the protocol means sending to all nodes.

\subsection*{Initiating Consensus}

\begin{algorithm}[t]\label{algo:initiate}
\If{$v$ initiates consensus at time $t$}{broadcast init$(H_v(t))$}
\If{received \textnormal{init}$(H_w)$ from $w\in V$ at
time $t$ and $|H_w-M_{vww}(t)|\leq 3\vartheta d$}{broadcast echo($w,H_w$)}
\nllabel{line:echo}
\If{received \textnormal{echo}$(w,H_w)$ from node $u$ at time
$t$ and $|H_w-M_{vww}(t)|\leq \Delta$\nllabel{line:store_cond}}{
  store $(u,\,$echo$(w,H_w))$\\
  \If{$|\{u\in V\,|\,(u,\,$\textnormal{echo}$(w,H_w))\mbox{ stored}\}|\geq f+1$
  and $E_v^w(H_w)$ is expired}{reset $E_v^w(H_w)$}
}
\If{$E_v^w(H_w)$ expires at time $t$}{
  \If{$|\{$stored tuples $(\cdot,\cdot,\textnormal{echo}
  (w,H_w))\}| \geq n-f$}{participate in $(w,H_w)$ with input
  $i_v(w,H_w,t)$} \nllabel{line:trust}
  \Else{participate in $(w,H_w)$ with input $0$}
}
\caption{Actions of node $v\in V$ at time $t$ that relate to initiating
consensus instances. $\Delta$ is a sufficiently large constant that will be
fixed later.}
\end{algorithm}

\algref{algo:clock} forces Byzantine nodes to announce consistent clock values
to most of the correct nodes or be revealed as faulty. In particular, it is not
possible for a Byzantine node to convince two correct nodes to accept
significantly different estimates of its clock.

However, timestamps alone are insufficient to guarantee the consistency of every
execution of the consensus protocol. Even if correct nodes know that a node
claiming to initiate consensus is faulty, they might be forced to participate in
the respective instance because unsuspecting nodes require the assistance of
\emph{all} correct nodes to overcome $f<n/3$ faults. Ironically, it would
require to solve agreement in order for all correct nodes to either participate
or not. This chicken-and-egg problem can be avoided using a gradecast-like
technique, cf.~\cite{feldman97}. If at least $n-f$ nodes send an echo message
(supposedly in response to an initiate message) in a timely fashion
(corresponding to confidence level 2 in gradecast), the initiating node might be
correct. Hence the receiver $w$ participates in the respective instance, with
input determined by $i_w$. If between $f+1$ and $n-f-1$ echo messages are
received (confidence level 1), the node participates (as there might be a
correct node that fully trusts in the instance), but defaults its input value to
``0''. Finally, if $f$ or less echo messages are received (confidence level 0),
it is for sure that no correct node participates with non-zero input and it is
safe to ignore the instance.

For every $w\in V\setminus \{v\}$, $v$ has a timeout $E_v^w(H)$, $H\in \R^+_0$,
of duration $2\vartheta d$, which serves to delay the start of an instance until
all nodes had time to make their decision. \algref{algo:initiate} gives the
pseudocode of the subroutine. We will choose $\Delta$ sufficiently large such
that each correct node waits for all correct nodes' echoes before deciding which
input to use.

\begin{algorithm}[t!]\label{algo:run}
\If{$H_v(t)=H_v(t_v)$}{$H_v^{(1)}:=H_v(t_v)+C$}
\If{received message $(m,i)$ from $u\in V\setminus\{v\}$ at time $t$
  and no tuple $(u,\cdot,i)$ stored\nllabel{line:receive}}{
  store $(u,m,i)$\\
  \If{$|\{(u,m,i)\,|\,(u,m,i)\mbox{ stored}\}|\geq
  n-f$ and $H_v^{(i+1)}=\bot$}
  {$H_v^{(i+1)}:=H_v(t)+2\vartheta d$}\nllabel{line:next}
  \If{$|\{(u,m,i)\,|\,(u,m,i)\mbox{ stored}\}|\geq f+1$
  and $(H_v^{(i)}=\bot $ or $H_v^{(i)}>H_v(t))$}{$H_v^{(i)}:=H_v(t)$}
  \nllabel{line:catch_up}}
\If{$H_v(t)=H_v^{(1)}$}{compute $M_v^{(1)}$ based on
input}\nllabel{line:first}
\If{$H_v(t)=H_v^{(i+1)}$ for $i\leq R-1$}{%
  compute $M_v^{(i+1)}$, where $\exists$ stored tuple $(u,m,i)$ with $m\neq
  \emptyset \Leftrightarrow$ received $m$ from $u$ in round $i$
  \nllabel{line:compute}
}
\If{$H_v(t)=H_v^{(i)}$ for $i\leq R$}{
  \For{$w\in V$}{
    \If{$\exists (m,w)\in M_v^{(i)}$}{send $(m,i)$ to $w$}
    \Else{send $(\emptyset,i)$ to $w$}
  }
}
\If{$H_v(t)=H_v^{(R+1)}$}{
  compute output, where $\exists$ stored tuple $(u,m,R)$ with $m\neq \emptyset
  \Leftrightarrow$ received $m$ from $u$ in round $R$
  \nllabel{line:output}
}
\caption{Actions of $v\in V$ at time $t$ that relate to running instance
$(w,H_w)$ invoked at time $t_v$. $C$ is a sufficiently large constant that will
be fixed later.}
\end{algorithm}

\subsection*{Running Consensus}

Denote by $t_v$ the time when $v$ decides to participate in instance $(w,H_w)$,
and by $M_v^{(i)}$, $i\in \{1,\ldots,R\}$, the messages it needs to send in
round $i$ of the protocol. Note that since \algref{algo:initiate} also specifies
node $v$'s input, it can compute $M_v^{(1)}$ (the messages to send in the first
round of the simulated consensus algorithm) by time $H_v^{(1)}:=H_v(t_v)+C$,
where $C$ is a suitable constant that will be specified later. All messages of
the instance are labelled by $(w,H_w)$ in order to distinguish between
instances. For ease of notation, we omitted these labels in \algref{algo:run}.

Essentially, the algorithm runs the fault-tolerant synchronization algorithm
from \cite{srikanth87} to ensure that the clock drift does not separate the
nodes' estimates of the progression of time during the execution by too much. If
a node can be sure that some correct node performed round $i$ (because it
received $f+1$ corresponding messages), it knows that it can safely do so
himself. To progress to the next round, nodes wait for $n-f$ nodes. Of these
$n-2f\geq f+1$ must be correct and will make sure that others catch up. A
timeout of $2\vartheta d$ guarantees that this information spreads and all
messages of round $i$ can be received before round $i+1$ actually starts. The
``non-messages'' $\emptyset$ are explicitly sent to compensate for missing
messages.

Note that if not all correct nodes participate, the timing bounds stated
above may become violated. However, since the employed protocol is silent and we
made sure that all inputs are $0$ if not all correct nodes participate,
interpreting missing messages as no message being received is sufficient to
ensure a consistent execution outputting $0$ at all nodes in this case.

\section{Analysis}\label{sec:analysis}

\subsection*{Distributing Clocks}

As mentioned earlier, we do not have to worry about correct initialization here,
since the ultimate goal is a self-stabilizing algorithm. To simplify the
following analysis, we may thus assume that at time $0$ each node sends two
consecutive (imagined) zero-delay update messages. This avoids issues in the
proof logic when referring to previous such messages.

First, we show that correct nodes maintain trusted and accurate clock estimates
of each other.
\begin{lemma}\label{lemma:clock_good}
If $v,w\in G$, then at any time $t$ it holds that
$H_w(t)\geq M_{vww}(t)\geq H_w(t)-3\vartheta d$.
\end{lemma}
\begin{proof}
Node $w$ sends a clock update at least every $2\vartheta$ local time. Since
messages are delayed by at most $d$ time units, the clock of $w$ will proceed by
at most $\vartheta d$ until such a message is received. Recall that we assume
that $M_{vww}(0)=H_w(0)$. Thus, it is sufficient to show that $w$ never sets
$M_{vww}:=\bot$, implying that it always sets $M_{vww}$ to a value from
$(H_w(t)-3\vartheta d,H_w(t)]$ before $M_{vww}(t)=H_w(t)-3\vartheta d$ becomes
satisfied.

Assume for contradiction that $t$ is the minimal time when some node $v\in G$
sets $M_{vww}:=\bot$ for some node $w\in G$. The clock of $v$ proceeds by at
most $(2\vartheta+1)d$ between consecutive updates from $w$. Together with the
assumption that $R_v^w(0)=H_v(0)$, this shows $v$ cannot execute
\lineref{line:too_slow} of \algref{algo:clock} at time $t$. Similarly, since
nodes send clock updates every $2\vartheta d$ local time (i.e., at most $2d$
real time apart) and messages are delayed by at most $d$, $v$ cannot set
$M_{vww}:=\bot$ according to \lineref{line:too_fast} of the algorithm at time
$t$. This leaves \lineref{line:check_support} as remaining possibility. We claim
that $|M_{vww}-M_{vuw}|\leq H_w(t)-H_w(0)\leq (2\vartheta^2+4\vartheta) d$ for
all $u\in G$. Given that $|G|\geq n-f$, from this claim we can conclude that $v$
does not execute this line at time $t$ either, resulting in a contradiction.

Consider the most recent update message (before time $t$) $w$ received from a
node $u\in G$. It has been sent at a time $t_u\geq t-2\vartheta d+d$, as
otherwise the next update message would already have arrived. Since $t>t_u$ is
minimal, we have that $H_u(t_u)\geq M_{uww}(t_u)\geq H_u(t_u)-3\vartheta d$. We
conclude that
\begin{equation*}
|M_{vww}(t)-M_{vwu}(t)|\leq |H_w(t)-H_w(t_u)|+|H_w(t_u)-M_{uww}(t_u)|\leq 
(2\vartheta^2+4\vartheta)d,
\end{equation*}
as claimed. By the previous observations, this completes the proof.
\end{proof}

The next lemma shows that the employed consistency checks force faulty nodes to
present reasonably similar clock estimates to different correct nodes.
\begin{lemma}\label{lemma:clock_bad}
Suppose that $v,w\in G$, $u\in V$, $t_w\geq t_v$, and
$M_{vuu}(t_v)\neq \bot \neq M_{wuu}(t_w)$. Then it holds that
$M_{wuu}(t_w)-M_{vuu}(t_v)\in 
[2(t_w-t_v)/(2\vartheta+3)-\BO(d),2\vartheta(t_w-t_v)+\BO(d)]$.
\end{lemma}
\begin{proof}
Consider the most recent update messages $v$ and $w$ received until time $t_v$,
at times $t_v',t_w'\in (t_v-(2\vartheta+1)d,t_v]$. Due to the prerequisites that
$M_{vuu}(t_v)\neq \bot \neq M_{wuu}(t_w)$, neither does $v$ set $M_{vuu}:=\bot$
at time $t_v'$ nor does $w$ set $M_{wuu}:=\bot$ at time $t_w'$. Hence,
\begin{equation*}
\exists X_v\subseteq V: |X_v|\geq n-f \wedge \forall x\in X_v:
|M_{vuu}(t_v')-M_{vxu}(t_v')|\leq (2\vartheta^2+4\vartheta)d,
\end{equation*}
and there is a set $X_w$ satisfying the same condition for $w$ at time $t_w'$.
Clearly, $|X_v\cap X_w|\geq n-2f\geq f+1$. Hence, there is a correct node
$g\in X_v\cap X_w\cap G$.

Denote by $t_r^v,t_r^w\in (t_v-(2\vartheta+1)d,t_v]$ the receiving times of the
latest update messages from $g$ that $v$ and $w$ received until time $t_v$ and
by $t_s^v,t_s^w\in (t_v-(2\vartheta+1)d,t_v]$, respectively, their
sending times. Note that there never is more than one update message from $g$ in
transit. Therefore, either $t_s^v=t_s^w$ or one of the messages received by $v$
and $w$ directly precedes the other one. Thus, $|H_g(t_s^v)-H_g(t_s^v)|\leq
2\vartheta d$. Within $2\vartheta d$ time, $g$ receives at most $\lceil
2\vartheta\rceil$ update messages from $u$, each of which must increases its
estimate $M_{guu}$ of $u$'s clock by exactly $2\vartheta d$, as otherwise it
would set $M_{guu}:=\bot$. We conclude that $|M_{guu}(t_s^v)-M_{guu}(t_s^w)|\in
\BO(d)$, yielding
\begin{eqnarray*}
&&|M_{vuu}(t_r^v)-M_{wuu}(t_r^w)|\\
&\leq &  |M_{vuu}(t_r^v)-M_{vgu}(t_r^v)|
+|M_{guu}(t_s^v)-M_{guu}(t_s^w)|
+|M_{wgu}(t_r^v)-M_{wuu}(t_r^w)| \in \BO(d).
\end{eqnarray*}

It remains to bound the progress of the estimates $M_{vuu}$ and $M_{wuu}$ during
$[t_r^v,t_v]$ and $[t_r^w,t_w]$, respectively. Again, $v$ must not set
$M_{vuu}:=\bot$ during $[t_r^v,t_v]$ and $w$ must not set $M_{wuu}:=\bot$
during $[t_r^w,t_w]$. Due to the fact that $v$ and $w$ accept update messages
without losing trust in $u$ only if they arrive at least $d$ and at most
$(2\vartheta^2+3\vartheta)d$ time apart, we can bound
\begin{equation*}
\frac{2(t_v-t_r^v)}{2\vartheta+3}-2\vartheta d\leq
M_{vuu}(t_v)-M_{vuu}(t_r^v) \leq
2\vartheta(t_v-t_r^v)+2\vartheta d
\end{equation*}
and
\begin{equation*}
\frac{2(t_w-t_r^w)}{2\vartheta+3}-2\vartheta d\leq
M_{wuu}(t_w)-M_{vuu}(t_r^w) \leq
2\vartheta(t_w-t_r^w)+2\vartheta d
\end{equation*}

Putting all bounds together, we obtain
\begin{eqnarray*}
M_{wuu}(t_w)-M_{vuu}(t_v)
\!\!&=&\!\!  M_{wuu}(t_w)-M_{wuu}(t_r^w)+M_{wuu}(t_r^w)-M_{vuu}(t_r^v)
-(M_{vuu}(t_v)-M_{vuu}(t_r^v))\\
\!\!&\in&\!\!
\left[\frac{2(t_w-t_v)}{2\vartheta+3}-\BO(d), 2\vartheta(t_w-t_v)+\BO(d)\right],
\end{eqnarray*}
concluding the proof.
\end{proof}

\subsection*{Initiating Consensus}

Having set up the bounds on the differences of clock estimates among correct
nodes, we can discuss their mutual support in invoking consensus. First, we show
that correct nodes can initiate instances unimpaired by the consistency checks
of \algref{algo:initiate}.

\begin{lemma}\label{lemma:initiate_good}
If $v\in G$ initiates a consensus instance at time $t$ and $\Delta\geq
3\vartheta d$, then each node $w\in G$ participates at some time $t_w\in
[t+2d,t+\BO(d)]$ with input value $i_w(v,H_v(t),t_w)$.
\end{lemma}
\begin{proof}
By \lemmaref{lemma:clock_good}, we have for all times $t'\in [t,t+2d]$ and nodes
$w\in G$ that 
$H_v(t)-3\vartheta d \leq M_{wvv}(t) \leq M_{wvv}(t_w)\leq 
H_v(t_w)\leq H_v(t)+2\vartheta d$.
Each node $w\in G$ will receive the init($H_v(t)$) message from $v$ at some time
$t'\in [t,t+d)$ and, as by the above bound the condition in \lineref{line:echo}
is met, broadcast an echo($v,H_v(t)$) message. These messages will be received
at times $t'\in [t,t+2d)$ and, as the condition in \lineref{line:store_cond} is
met, be stored by nodes $w\in G$.

Since only faulty nodes may send an echo($v,H_v(t)$) message earlier than time
$t$ and $|G|\geq n-f > f+1$, the condition for resetting $E_w^v(H_v(t))$ will be
met at each $w\in G$ at some time during $[t,t+2d)$. Therefore, each such node
participates in the instance $(v,H_v(t))$ at some time $t_w\in [t+2d,t+\BO(d)]$.
By this time, $w$ will have received all echo($v,H_v(t)$) messages from nodes in $G$. Thus,
the condition in \lineref{line:trust} is met at time $t_w$ and it will use input
$f_w(t_w)$.
\end{proof}

The following statement summarizes how the guarantees of the clock estimates
control faulty nodes' ability to feed inconsistent information to correct
nodes by timing violations.
\begin{corollary}\label{coro:echo}
If at times $t_v,t_w\in \R^+_0$ nodes $v,w\in G$ send echo$(u,H_u)$, then
$|t_v-t_w|\in \BO(d)$.
\end{corollary}
\begin{proof}
By \lineref{line:echo} of \algref{algo:initiate}, we have that
$|M_{vuu}(t_v)-M_{wuu}(t_w)|\leq 6\vartheta d$. By \lemmaref{lemma:clock_bad},
$|M_{vuu}(t_v)-M_{wuu}(t_w)|\in \Omega(|t_v-t_w|)$. Hence, $|t_v-t_w|\in
\BO(d)$.
\end{proof}

This entails that correct nodes use non-zero input only when all correct nodes
participate.

\begin{lemma}\label{lemma:synchronous_start}
Suppose that $\Delta\in \BO(d)$ is sufficiently large. If for any $u\in V$,
$v\in G$ participates in a consensus instance labeled $(u,H_u)$ with an input
value different from $0$ at time $t_v$, then each node $w\in G$ participates at
some time $t_w\in [t^-,t^+]$, where $t^+-t^-\in \BO(d)$.
\end{lemma}
\begin{proof}
Since $v$ participates in the instance with non-zero input, it stores at least
$n-f$ tuples $(x,$\,echo$(u,H_u))$. At least $n-2f\geq f+1$ of these correspond to
echo$(u,H_u)$ messages sent by correct nodes. Since $v$ participates at time
$t_v$, it received one of these messages at some time $t_v-\Theta(d)$. By
\corollaryref{coro:echo}, all such messages sent by correct nodes must have
been sent (and thus received) within an interval $[t_v-\BO(d),t_v+\BO(d)]$. We
conclude that (i) no correct node will join the instance earlier than time
$t_v-\BO(d)$, (ii) all correct nodes will receive at least $f+1$
echo$(u,H_u)$ messages from different sources by time $t_v+\BO(d)$, (iii) as
$\Delta$ is sufficiently large, at all correct nodes the condition in
\lineref{line:store_cond} of \algref{algo:initiate} will be met when receiving
these messages, and therefore (iv) all correct nodes join the instance by time
$t_v+\BO(d)$.
\end{proof}

\subsection*{Running Consensus}

The silence property of the employed consensus protocol deals with all instances
without a correct node with non-zero input. \lemmaref{lemma:synchronous_start}
shows that all correct nodes participate in any other instance. Hence, we need
to show that any instance in which all correct nodes participate successfully
simulates a synchronous execution of the consensus protocol.

\begin{lemma}\label{lemma:consensus_correct}
Suppose that $\Delta,C\in \BO(d)$ are sufficiently large and that some node from
$G$ participates in instance $(v,H_v)$ at time $t_0$ with input value different
from $0$. Then each node $w\in G$ computes an output for the instance
(\lineref{line:output} of \algref{algo:run}) at some time $t_w\in [t^-,t^+]$,
where $t^+-t^-\in \BO(d)$ and $t^-,t^+\in t_0+\Theta(R)$. These outputs are the
result of some synchronous run of ${\cal P}_{\!s}$ with the inputs the nodes
computed when joining the instance.
\end{lemma}
\begin{proof}
We will denote for each node $w\in G$ and each $i\in \{1,\ldots,R+1\}$ by
$t_w^{(i)}$ the time satisfying that $H_w(t_w^{(i)})=H_w^{(i)}(t_w^{(i)})$; we
will show by induction that these times exist and are unique. Define
$t^{(i)}:=\min_{w\in G}\{t_w^{(i)}\}$. The induction will also show that all
nodes $w\in G$ compute and send their messages, as well as receive and store all
messages from other nodes in $G$ for rounds $j<i$, $i\in \{2,\ldots,R+1\}$, of
the protocol (i.e., execute Lines \ref{line:first} or \ref{line:compute} and
\ref{line:receive} of \algref{algo:run}) at times smaller than $t^{(i)}$. Note
that these properties show that the progression of \algref{algo:run} can be
mapped to a synchronous execution of ${\cal P}_{\!s}$ and the messages
$M_w^{(i)}$ can indeed be computed according to ${\cal P}_{\!s}$. Finally, the
induction will show that $t_w^{(i)}\in t_0+\Theta(id)$ for all $w\in G$ and
$i\in \{2,\ldots,R+1\}$; the stated time bounds on $t^-$ and $t^+$ follow. As
the messages $M_w^{(1)}$ the nodes compute in \lineref{line:first} are
based on the inputs the node compute when joining the instance, completing the
induction will thus also complete the proof.

Before we perform the induction, let us make a few observations. The only way to
manipulate $H_w^{(i)}\neq \bot$ at some time $t$ is to set it to $H_w(t)$,
provided it was larger than that (\lineref{line:catch_up}). Thus, once defined,
$H_w^{(i)}(\cdot)$ is non-increasing, and can never be set to a value smaller
than $H_w(t)$. In particular, the times $t_w^{(i)}$ are unique (if they exist).
Furthermore, the conditions for computing and sending messages are checked after
this line, implying that the lines in which messages are computed and sent are
indeed performed at the unique time $t_w^{(i)}$. Therefore, each node $u\in G$
sends (at most) one message $(\cdot,i)$ to each node $w\in G$, which will be
received and stored a time from $(t_w^{(i)},t_w^{(i)}+d)$, assuming that the
receiver already joined the instance. The latter can be seen as follows. We
apply \lemmaref{lemma:synchronous_start} to see that each node $w\in G$
participates in the instance at some time $t_w^{(0)}\in t_0+\Theta(d)$. Thus, if
$C\in \BO(d)$ is sufficiently large, each node $w\in G$ has joined the instance
and computed $H_w^{(1)}=H_w(t_w^{(0)})+C$ before time $t^{(1)}$ (which exists
because $H_w^{(1)}$ has been set to some value).

We now perform the induction step from $i\in \{1,\ldots,R\}$ to $i+1$. First,
let us show that the times $t_w^{(i+1)}$ exist. Since each node $G$ sends some
message $(m,i)$ to each other node in $G$ at some time from $t_0+\Theta(i d)$,
each node $w\in G$ will execute \lineref{line:next} for $i$ at some time $t\in
t_0+\Theta(i d)$, setting $H_w^{(i+1)}:=H_w(t)+2\vartheta d$. We conclude that
the times $t_w^{(i+1)}$ exist. Clearly, no node in $G$ can execute
\lineref{line:catch_up} before time $t^{(i+1)}$, as until then no messages
$(\cdot,i+1)$ are sent by any nodes in $G$. Hence, $t_w^{(i+1)}\in
t_0+\Theta((i+1)d)$ for all $w\in G$. Now suppose that $t_{i+1}$ is minimal with
the property that some node $w\in G$ executes \lineref{line:next}, defining
$H_w^{(i+1)}$. At this time, it stores $n-f$ tuples $(u,m,i)$ for $u\in V$, at
least $n-2f\geq f+1$ of which satisfy that $u\in G$. For each such $u\in G$, it
holds that $t_u^{(i)}<t_{i+1}$, implying that at each node in $x\in G$ the first
part of the condition for executing \lineref{line:catch_up} for index $i$ will
be satisfied at some time smaller than $t_{i+1}+d$. Consequently,
$t_x^{(i)}<t_{i+1}+d$, and all messages from nodes in $G$ corresponding to round
$i$ will be sent by time $t_{i+1}+d$ and received by time $t_{i+1}+2d\leq
t^{(i+1)}$. By induction hypothesis, the same holds for all messages to and from
nodes in $G$ for rounds $j<i$. Thus, all claimed properties are satisfied for
step $i+1$, completing the induction and hence the proof.
\end{proof}

We conclude that Algorithms~\ref{algo:clock}--\ref{algo:run} together solve the
initialization problem.
\begin{theorem}\label{theorem:simulation}
Each consensus instance $(v,H_v)$ can be mapped to a synchronous execution of
${\cal P}_{\!s}$. If the instance has output $o\neq 0$, all nodes in $G$
output $o$ within $\BO(d)$ time of each other. Moreover, there is a node
in $w\in G$ satisfying that $f_w(t)=o$ for some time $t\in t_w-\Theta(R d)$,
where $t_w$ is the time when it outputs $o$. Finally, if $v$ is correct and
$t_0$ is the time when it initialized the instance, all nodes in $w\in G$
compute their inputs as $f_w(t_w)$ at some time $t_w\in t_0+\Theta(d)$. 
\end{theorem}
\begin{proof}
Assume first that no node in $G$ participates in the instance with an input
different from $0$. Then no node in $G$ will send a message $(m,i)$ for any $i$
with $m\neq \emptyset$ for this instance: ${\cal P}_{\!s}$ is silent, and
\algref{algo:run} interprets any ``missing'' message as having received no
message from the respective node in Lines \ref{line:compute} and
\ref{line:output}; in particular, all nodes in $G$ will output $0$.

Next, suppose that some correct node has input different from $0$. In
this case, the claimed properties follow from \lemmaref{lemma:consensus_correct}
and the properties of ${\cal P}_{\!s}$.

Finally, assume that $v\in G$ and $t_0$ is the time when $v$ initializes the
instance. \lemmaref{lemma:initiate_good} shows that each node $w\in G$
participates in the instance at some time $t_w\in t_0+\Theta(d)$ with input
$f_w(t_w)$.
\end{proof}

With the initialization problem being solved, it is straightforward to derive an
algorithm that enables consistent initialization of arbitrary consensus
protocols.
\begin{corollary}\label{coro:simulation}
Given any $R$-round synchronous consensus algorithm ${\cal P}$ tolerating
$f<n/3$ faults, there is an algorithm with the following guarantees.
\begin{compactitem}
  \item Each (correct) node can initiate an instance of ${\cal P}$ at any time
  $t$.
  \item For any instance (also those initiated by faulty nodes) it holds that
  nodes determine their inputs according to their local view of the system
  during some interval $[t_1,t_1+\BO(1)]$, and terminate during some interval
  $[t_2,t_2+\BO(1)]$, where $t_2\in t_1+\Theta(R)$.
  \item If a correct node initiates an instance at time $t$, then $t_1=t$.
  \item Each instance satisfies termination, agreement, and validity.
  \item The above guarantees hold in the presence of $f$ faulty nodes.
\end{compactitem}
\end{corollary}
\begin{proof}
We run algorithms Algorithms~\ref{algo:clock}--\ref{algo:run} in the background,
with input functions always returning $1$ and ${\cal P}_{\!s}$ (the derived silent
protocol from \lemmaref{lemma:silent}) as the utilized silent consensus
protocol. Whenever a node wants to initiate an instance of ${\cal P}$ at a time
$t$, it first initiates an instance of ${\cal P}_{\!s}$ using our framework. When
reaching the threshold of echo messages to participate in the instance (at some
time from $(t,t+\BO(d))$), correct nodes store the input they will use if this
call leads to an actual run of ${\cal P}$, according to their current view of
the system.

Provided that a correct node initiates an instance, by
\theoremref{theorem:simulation} all correct nodes will compute output $1$ for
the associated instance of ${\cal P}_{\!s}$ (by validity). This is mapped to
starting an associated run of ${\cal P}$ with the inputs memorized earlier,
where a copy of \algref{algo:run} is used to run ${\cal P}$. Note that, since
all correct nodes participate, \lemmaref{lemma:consensus_correct} shows that we can
map the execution of \algref{algo:run} to a synchronous execution of ${\cal P}$
with the inputs determined upon initialization, where each correct node
terminates during an interval $[t',t'+\BO(d)]$ for some $t'\in t+\Theta(R)$.

On the other hand, output $0$ is mapped to taking no action at all. For
instances of ${\cal P}_{\!s}$ that output $1$, \theoremref{theorem:simulation}
shows that all nodes terminate within $\BO(d)$ time off each other. Previous
arguments also show that the inputs to the resulting run of ${\cal P}$ have been
determined $\Theta(R)$ time earlier, as desired. We conclude that all claimed
properties are satisfied.
\end{proof}

\section{Self-Stabilization and Bounded Communication
Complexity}\label{sec:self-stab}

In this section, we discuss how the previous results can be generalized to
\theoremref{theorem:main}. We will add self-stabilization first, then argue how
to use discrete and bounded clocks, and finally control the rate at which
consensus instances can be initiated.

\subsection*{Self-Stabilization}

Within $d$ time, the links deliver all spurious messages from earlier times;
afterwards, each message received from a correct node will be sent in accordance
with the protocol.

We take a look at the individual components of the algorithm.
\algref{algo:clock} is not self-stabilizing, because the loss of trust in a node
cannot be reversed. This is straightforward to rectify, by nodes starting to
forward received claimed clock values if their senders are well-behaving for
sufficient time, and subsequently starting to trust a node again if receiving
consistent reports on its clock from $n-f$ nodes for sufficiently long. This is
detailed in \sectionref{app:clocks}, where we present
\algref{algo:clock_stab}, a self-stabilizing variant of \algref{algo:clock}.

As \algref{algo:clock_stab} will operate correctly after $\BO(R)$ time, it is
not hard to see how to make \algref{algo:initiate} self-stabilizing. We know
that a ``correct'' execution for a given label will start with a ``clean slate''
(i.e., no tuples stored at any correct node). All related messages sent and
received by correct nodes as well as possibly joining the instance are confined
within a time window of length $\tau\in \BO(d)$. Hence, we can add timeouts
deleting stored tuples from memory $\vartheta\tau$ local time after they have
been written to memory, without disturbing the operation of the
algorithm.\footnote{Note that this can be done in a self-stabilizing way by
memorizing the local times when they have been stored; if such a time lies in
the future or more than $\vartheta\tau$ time in the past (according to the
current value of the hardware clock), the entries need to be deleted.} By making
the time to regain trust in a (faulty) node's clock (distributed by
\algref{algo:clock_stab}) larger than $\vartheta\tau$, we can guarantee that
memory will be wiped before the faulty node can ``reuse'' the same label at a
later time (by ``setting its hardware clock back''). This modification ensures
that \algref{algo:initiate} will stabilize within $\BO(d)$ time once
\algref{algo:clock_stab} does.

Similar considerations apply to \algref{algo:run}. We know that a ``correct''
execution of the algorithm progresses to the next simulated round of ${\cal
P}_{\!s}$ within $\tau\in\BO(d)$ time (all correct nodes participate) or correct
nodes do not send any messages in the simulated execution of ${\cal P}_{\!s}$
and output $0$ (by silence). Adding a timeout of $\vartheta \tau$ (locally)
terminating the instance with output $0$ if no progress is made thus guarantees
termination within $\BO(R)$ time. Naturally, this may entail that correct nodes
``leave'' an instance prematurely, but this may happen if the instance was not
initialized correctly (i.e., nodes have lingering false memory entries from time
$0$) or the instance is silent (i.e., there is no need to send messages and the
output is $0$ at all correct nodes) only. Similar to \algref{algo:initiate}, this
strategy guarantees that false memory entries can be safely wiped within
$\vartheta \tau R$ rounds; increasing the timeout to regain trust in
\algref{algo:clock_stab} to $\vartheta^2 \tau R$ thus guarantees that
\algref{algo:run} will stabilize within $\BO(R)$ rounds once Algorithm
\ref{algo:clock_stab} and \ref{algo:initiate} have, in the sense that to its
future outputs the arguments and bounds from \sectionref{sec:analysis}
apply.

Finally, we note that when calling \algref{algo:run} for protocol ${\cal P}$ in
\corollaryref{coro:simulation}, always all nodes participate. Hence, the same
arguments apply and a total stabilization time of $\BO(R)$ follows.

\subsection*{Discrete and Bounded Clocks}

In practice, clocks are neither continuous nor unbounded; moreover, we need
clock values to be bounded and discrete to encode them using few bits.
Discretizing clocks with a granularity of $\Theta(d)$ will asymptotically have
no effect on the bounds: We simply interpret the discrete clocks as readings of
continuous clocks with error $\BO(d)$. It is not hard to see that this can be
mapped to a system with exact readings of continuous clocks and larger maximal
delay $d'\in \BO(d)$, where all events at node $v$ happen at times when
$H_v(t)\in \N$.\footnote{This entails that timeouts are integer, which also
clearly does not affect the asymptotic bounds.}

As shown in \corollaryref{coro:clock_bad}, choosing $B\in \Theta(R)$ in
\algref{algo:clock_stab} guarantees the following. For each sufficiently large
time $t\geq t_0\in \Theta(R)$, all correct nodes from $G$ trusting some node
$v\in V$ at time $t$ received clock values from $v$ that increased at constant
rate for $\Theta(R)$ time and differed at most by $\BO(d)$. From this it follows
that using clocks modulo $M\in \Theta(R)$ is sufficient: Choosing $M$
sufficiently large, we can make sure that for any label $(v,H)$, every
$\Theta(R)$ time there will be a period of at least $\vartheta^2\tau R$ ($\tau$
as above) time during which all correct node reject initialization messages
labeled $(v,H)$. This ensures that memory will be wiped before the next messages
are accepted and the previous arguments for self-stabilization apply.

\subsection*{Bounding the Communication Complexity}

Using bounded and discrete clocks and assuming that $R$ is polynomially bounded
in $n$, each clock estimate (and thus each label) can be encoded by $\BO(\log
n)$ bits. Hence, each correct node will broadcast $\BO(n\log n)$ bits in
$\Theta(d)$ time when executing \algref{algo:clock_stab}, for a total of
$\BO(n^2\log n)$ bits per node and time unit.

However, so far each node may initiate an instance at any time, implying that
faulty nodes could initiate a large number of instances with the goal of
overloading the communication network. Hence, we require that correct nodes wait
for at least $T\geq 2\vartheta d$ local time between initializing instances.
Under this constraint, it is feasible that correct nodes ignore any init message
from $v\in V$ that is received less than $T/\vartheta-d$ local time after the
most recent init message from $v$. As a result, no node will broadcast more than
$\BO(n \log n)$ bits within $T$ time due to executing \algref{algo:initiate}.

Moreover, now there cannot be more than one instance per node $v$ and
$\tilde{T}=(T/\vartheta-d)/\vartheta$ time such that some correct node
participates with non-zero input due to messages sent at times greater than $0$
alone (i.e., not due to falsely memorized echo messages at time $0$): this
requires the reception of $n-2f$ corresponding echo messages from correct nodes,
which will not send echo messages for another instance labeled $(v,\cdot)$ for
$T/\vartheta$ time. Such an instance runs for $\BO(R)$ time. There are at most
$|G|=n-f$ other instances with label $(v,\cdot)$ a node may participate in
within $\tilde{T}$ time ($f+1$ received messages imply one was from a correct
node), all of which terminate within $2$ simulated rounds with ``empty''
messages $(\emptyset,1)$ or $(\emptyset,2)$ only.

For any $v\in V$, this leads to the following crucial observations: (i) If a
node memorizes that it participates in more than $k_1\in \BO(R/\tilde{T})$
instances labeled $(v,\cdot)$ which did not terminate by the end of round $2$ or
sent other messages than $(\emptyset,1)$ or $(\emptyset,2)$, its memory content
is inconsistent; (ii) if a node memorizes that it participates in more than
$k_2\in\BO(n/\tilde{T})$ instances $(v,\cdot)$, its memory content is
inconsistent; (iii) as memorized echo messages and memory associated with an
instance of \algref{algo:run} is cleared within $\BO(R+T)$ time, (i) or (ii) may
occur at times $t\in \BO(R+T)$ only; and (iv) if a node $w\in G$ detects (i) or
(ii) at time $t$ and deletes at time $t+d$ all memorized echo messages, forces all
timeouts $E_v^{\cdot}(\cdot)$ into the expired state, and clears all memory
entries corresponding to \algref{algo:run}, (i) or (ii) cannot happen again at
this node.

Hence, we add the rule that a node detecting (i) or (ii) stops sending any
messages corresponding to \algref{algo:run} for $\vartheta d$ local time and
then clears memory according to observation (iv). By (iii), this mechanism will
stop interfering with stabilization after $\BO(R+T)$ time; afterwards, the
previous arguments apply. Furthermore, (i) and (ii) imply that a node never
concurrently participates in more than $k_1$ instances for which it sends
non-empty messages, and sends at most $\BO(n^2\log n)$ bits ($\BO(n)$
broadcasted round numbers and labels) in $\BO(\tilde{T})=\BO(T)$ time due to
other instances.

Hence, it remains to control the number of bits sent by the at most
$k_1$ remaining instances. Recall that the messages sent by
\algref{algo:run} are of the form $(m,i)$, where $m$ is a message sent by ${\cal
P}_{\!s}$ and $i$ is a round number. We know that in a correct simulated
execution of such an instance, the node sends up to $B+\BO(rn\log n)$ bits
within $rd$ time: $B$ is the maximal number of bits sent by a node in an
execution ${\cal P}$, the additional two initial round of ${\cal P}_{\!s}$
require nodes to broadcast single-bit messages, and $\log R\in \BO(\log n)$
broadcasted bits are required to encode round numbers and labels. Therefore, a
node can safely locally terminate any instance violating these bounds and
output, say, $0$. Such a violation may only happen if the instance has not been
properly initialized; since \emph{any} instance terminates within $\BO(R)$ time
and \algref{algo:clock_stab} and subsequently \algref{algo:initiate} will
stabilize within $\BO(R+T)$ time, we can conclude that, again, this mechanism
will not interfere with stabilization once $\BO(R+T)$ time has passed.

In summary, we have shown the following.
\begin{compactitem}
  \item We can modify the algorithm from \corollaryref{coro:simulation} such
  that it self-stabilizes in $\BO(R)$ time.
  \item We can further modify it to operate with bounded and discrete hardware
  clocks.
  \item For $T\geq 2\vartheta d$, additional modifications ensure that, for each
  correct node, the amortized number of bits sent per time unit is $\BO(n^2\log
  n+k_1nB)=\BO(n^2\log n+nBR/T)$; this increases the stabilization time to
  $\BO(R+T)$ and entails that correct nodes wait at least $T$ local time between
  initializing instances.
\end{compactitem}
The resulting statement is exactly \theoremref{theorem:main}.

\section{Self-Stabilizing Clock Distribution}\label{app:clocks}

\algref{algo:clock_stab}, the self-stabilizing variant of \algref{algo:clock},
is essentially identical, except that the loss of trust upon detecting an
inconsistency is only temporary. To this end node $v\in V$ maintains timeouts
$A_v^w$ and $B_v^w$ for each node $w\in V$, of durations $2\vartheta d$
and $B$, respectively. The clock estimate $v$ has of $w$ then is $M_{vww}(t)$ at
times $t$ when $B_v^w$ is expired and $\bot$ otherwise. Timeout $A_v^w$ is reset
whenever $w$ announces clock values to $v$ that violate the timing constraints,
i.e., an update message is sent too soon or too late after the previous, or it
does not have contain the previous value increased by $2\vartheta d$. Whenever
$A_v^w$ is not expired, $v$ will report $\bot$ as the ``clock value'' it
received from $w$ to others, expressing that there has been an inconsistency; at
other times, it reports the most recent value received. If a node keeps sending
values in accordance with the timing constraints, eventually all correct nodes
will be reporting these values (as their $A$-timeouts expire). Subsequently the
check in \lineref{line:reset_B} will always be passed, which resets $B_v^w$
whenever there is insufficient support from others for the clock value $w$
claims to $v$. Eventually, $B_v^w$ will expire, and $w$'s trust in $v$ is
restored.

\begin{algorithm}[ht!]
\If{$H_v(t)\!\!\mod 2\vartheta d = 0$}{
  \For{$w\in\{1,\ldots,n\}\setminus \{v\}$}{
    \If{$H_v(t)-R_v^w>(2\vartheta^2+\vartheta)d$}{reset $A_v^w$ and
    $B_v^w$\nllabel{line:reset_A_B_1}}
    \If{$A_v^w=1$}{$\hat{M}_{vww}:=M_{vww}$}
    \Else{$\hat{M}_{vww}:=\bot$}
  }
  $M_{vv}:=(\hat{M}_{v11},\ldots,\hat{M}_{vnn})$\nllabel{line:update-mvv}\\
  broadcast update$(M_{vv})$\nllabel{line:update_stab}
}
\If{received \emph{update}$(M_{w11},\ldots,M_{wnn})$ from node $w$ at time $t$}{
  \If{$H_v(t)-R_v^w< d$ or $M_{www}-M_{vww}\neq 2\vartheta
    d$\nllabel{line:check_reset_A_B_2}}{reset $A_v^w$ and
    $B_v^w$\nllabel{line:reset_A_B_2}}
  $M_{vw}:=(M_{w11},\ldots,M_{wnn})$\\
  \For{$x\in\{1,\ldots,n\}\setminus \{v\}$}{
    \If{$|\{u\in V\,|\,|M_{vxx}-M_{vux}|\leq
    (2\vartheta^2+4\vartheta)d\}|<n-f$}{reset $B_v^x$\nllabel{line:reset_B}}
  }
  $R_v^w:=H_v(t)$\nllabel{line:reset_R_stab}
}
\caption{Actions of node $v\in V$ at time $t$ that relate to maintaining
self-stabilizing clock estimates.}\label{algo:clock_stab}
\end{algorithm}

Note that it is straightforward to adapt the algorithm to bounded clocks
modulo some value $M\gg B$. As we just argued why the algorithm stabilizes in
the sense that correct nodes eventually trust each other, the following
analogon to \lemmaref{lemma:clock_good} is immediate.
\begin{corollary}\label{coro:clock_good}
Suppose that $t_0\in \BO(d+B)$ is sufficiently large. If $v,w\in G$, then at any
time $t\geq t_0$ it holds that $H_w(t)\geq M_{vww}(t)\geq H_w(t)-3\vartheta d$.
\end{corollary}
\lemmaref{lemma:clock_bad} is translated in a similar fashion.
\begin{corollary}\label{coro:clock_bad}
Suppose that $v,w\in G$, $u\in V$, $t_v\geq t_0$ for a sufficiently large
$t_0 \in \BO(d)$, $t_w\in [t_v,t_v+B/\vartheta-(2\vartheta+1)d]$, $B_v^u$ is
expired at time $t_v$, and $B_w^u$ is expired at time $t_w$. Then
$M_{wuu}(t_w)-M_{vuu}(t_v)\in
[2(t_w-t_v)/(2\vartheta+3)-\BO(d),2\vartheta(t_w-t_v)+\BO(d)]$.
\end{corollary}
\begin{proof}
The requirement that $t_v\geq t_0$ ensures that all spurious messages in the
communication network at time $0$ have been received and, afterwards, all
correct nodes sent and received at least two update messages from each other
correct node.

Consider the most recent update messages $v$ and $w$ received until time $t_v$,
at times $t_v',t_w'\in (t_v-(2\vartheta+1)d,t_v]$. Due to the prerequisites that
$B_v^u$ is expired at time $t_v$ and $B_w^u$ is expired at time $t_w$, neither
does $v$ set $M_{vuu}:=\bot$ at time $t_v'$ nor does $w$ set $M_{wuu}:=\bot$ at
time $t_w'$. Hence,
\begin{equation*}
\exists X_v\subseteq V: |X_v|\geq n-f \wedge \forall x\in X_v:
|M_{vuu}(t_v')-M_{vxu}(t_v')|\leq (2\vartheta^2+4\vartheta)d,
\end{equation*}
and there is a set $X_w$ satisfying the same condition for $w$ at time $t_w'$.
Clearly, $|X_v\cap X_w|\geq n-2f\geq f+1$. Hence, there is a correct node
$g\in X_v\cap X_w\cap G$.

From here we proceed analogously to the proof of \lemmaref{lemma:clock_good},
noting that $A_g^u$ being of duration $2\vartheta d$ guarantees that
$|M_{guu}(t_s^v)-M_{guu}(t_s^w)|\in \BO(d)$ for two consecutive update messages
sent by $g$ at times $t_s^v$ and $t_s^w$.
\end{proof}

\section{Further Results}\label{app:further}

One can reduce the bit complexity from \theoremref{theorem:main} further by
reducing the frequency at which clock estimates are updated. The loss in
accuracy however comes at the cost of increasing the time interval during which
input values are determined.
\begin{corollary}
Suppose ${\cal P}$ is a synchronous consensus protocol tolerating $f<n/3$
faults, runs for $R\in \operatorname{polylog}(n)$ rounds, and guarantees that no
correct node sends more than $B$ bits. For each $T\geq 2\vartheta d$, there is a
value $S\in \BO(T+R)$ and an algorithm with the following properties.
\begin{compactitem}
  \item Each correct node $v$ can initiate an instance of ${\cal P}$ at any time
  $t\geq S$, provided that it has not done so at any time $t'<t$ for which
  $H_v(t)-H_v(t')\leq T$.
  \item For any instance that terminates at a time larger than $S$, it holds
  that nodes determine their inputs according to their local view of the system
  during some interval $[t_1,t_1+\BO(T)]$, and terminate during some interval
  $[t_2,t_2+\BO(1)]$, where $t_2\in t_1+\Theta(T+R)$.
  \item If a correct node initiates an instance at time $t\geq S$, then $t_1=t$.
  \item Each instance for which $t_2\geq S$ satisfies termination, agreement,
  and validity.
  \item Each correct node sends at most $\BO((n^2\log n+RBn)/T)$ amortized bits
  per time unit.
  \item The above guarantees hold in the presence of $f$ faulty nodes and
  for arbitrary initial states.
\end{compactitem}
\end{corollary}
\begin{proof}
We apply our reasoning for $d'\in\Theta(T)$, except that \algref{algo:run} still
progresses at one simulated round within $\Theta(d)$ time.\footnote{Note that we
have to set $C\in \Theta(d')=\Theta(T)$, though.} In other words, nodes send
clock updates every $\Theta(T)$ time, implying that the clock estimates are
accurate up to $\Theta(T)$, and instances of \algref{algo:run} are joined within
a time window of $\Theta(T)$ by correct nodes. \algref{algo:run} thus terminates
within $\BO(C+R)=\BO(T+R)$ rounds, so we can choose the timeouts for regaining
trust in \algref{algo:clock_stab} and the maximal clock value in $\Theta(T+R)$
as well; this ensures that the stabilization time remains $\BO(T+R)$.

With these modifications, we have a bit complexity of $\BO(n^2\log n)$ per node
and $T$ time for Algorithms~\ref{algo:clock} and~\ref{algo:initiate}. The bound
of $\BO((n^2\log n+Bn)/T)$ amortized bits per node and time unit for
\algref{algo:run} holds as before, resulting in a total of $\BO((n^2\log
n+Bn)/T)$ bits per node and time unit for the compound algorithm.
\end{proof}

Since our framework is deterministic, it can operate in any adversarial model.
What is more, we make use of the agreement and validity properties of ${\cal P}$
only in executions simulating a synchronous execution of the protocol in which
all nodes participate. This happens only polynomially often in $n$. Hence, we
can also plug randomized consensus algorithms in our framework that satisfy
agreement and validity w.h.p.\ only. A randomized consensus protocol terminating
within $R$ rounds satisfies the following properties.
\begin{compactitem}
\item[\textbf{Termination:}] Every correct node terminates within $R$ rounds
and outputs a value $o(v)\in I$.
\item[\textbf{Agreement:}] With high probability, $o(v)=o(w)$ for correct nodes
$v,w$.
\item[\textbf{Validity:}] If $i(v)=i(w)$ for all correct $v,w$, with high
probability this is also the output value.
\end{compactitem}
Note that, typically, agreement and validity are required to hold
deterministically, while termination is only satisfied probabilistically.
It is simple to translate such an algorithm in one that satisfies the above
criteria by forcing termination after $R$ rounds, where $R$ is sufficiently
large to guarantee termination w.h.p.\footnote{Frequently, running time bounds
are shown to hold in expectation only. To the best of our knowledge, in all
these cases an additional factor of $\BO(\log n)$ is sufficient to obtain a
bound that holds w.h.p.} For suitable randomized algorithms, the following
corollary is immediate.
\begin{corollary}
Suppose ${\cal P}$ is a synchronous randomized consensus protocol tolerating
$f<n/3$ faults that terminates in $R\in \operatorname{polylog}(n)$ rounds and
guarantees that no correct node sends more than $B$ bits w.h.p. Then there is a
value $S\in \BO(R)$ and an algorithm with the following properties.
\begin{compactitem}
  \item Each correct node $v$ can initiate an instance of ${\cal P}$ at any time
  $t\geq S$, provided that it has not done so at any time $t'<t$ for which
  $H_v(t)-H_v(t')\leq R$.
  \item For any instance that terminates at a time larger than $S$, it holds
  that nodes determine their inputs according to their local view of the system
  during some interval $[t_1,t_1+\BO(R)]$, and terminate during some interval
  $[t_2,t_2+\BO(1)]$, where $t_2\in t_1+\Theta(R)$.
  \item If a correct node initiates an instance at time $t\geq S$, then $t_1=t$.
  \item Each instance for which $t_2\geq S$ satisfies agreement and validity
  w.h.p.\footnote{This statement holds \emph{per instance};
  during superpolynomially large time intervals, some instances may fail.}
  \item Each correct node sends at most $\BO(n^2\log n+Bn)$ bits within $R$
  time.
  \item The above guarantees hold in the presence of $f$ faulty nodes and
  for arbitrary initial states.
\end{compactitem}
\end{corollary}

\section*{Acknowledgements}
This material is based upon work supported by the National Science Foundation
under Grant Nos.\ CCF-AF-0937274, CNS-1035199, 0939370-CCF and CCF-1217506, the
AFOSR under Contract No.\ AFOSR Award number FA9550-13-1-0042, the Swiss Society
of Friends of the Weizmann Institute of Science, the German Research Foundation
(DFG, reference number Le 3107/1-1), the Israeli Centers of Research Excellence
(I-CORE) program, (Center  No.\ 4/11), grant 3/9778 of the Israeli Ministry of
Science and Technology, and the Google Inter-university center for ``Electronic
Markets and Auctions''. Danny Dolev is Incumbent of the Berthold Badler Chair.


\bibliographystyle{abbrv}
\bibliography{../log_clock}

\end{document}